\documentclass[11pt]{llncs}
\usepackage[left=1in,right=1in,top=1in,bottom=1in]{geometry}

\usepackage[english]{babel} 
\usepackage{amsmath,amssymb,amsfonts,mathrsfs}
\usepackage[T1]{fontenc}
\usepackage[utf8]{inputenc}
\usepackage{euscript}
\usepackage{dsfont}
\usepackage{booktabs}
\usepackage{url}

\usepackage{parskip}
\usepackage[ruled,vlined,linesnumbered]{algorithm2e}
\usepackage[hidelinks]{hyperref}
\usepackage{algpseudocode}

\spnewtheorem{assumption}{Assumption}{\bfseries}{\itshape}
\spnewtheorem*{lemma*}{Lemma}{\normalshape\bfseries}{\itshape}
\spnewtheorem*{theorem*}{Theorem}{\normalshape\bfseries}{\itshape}

\newcommand{\eqdef}{\triangleq}
\renewcommand{\leq}{\leqslant}

\renewcommand{\geq}{\geqslant}
\renewcommand{\ge}{\geqslant}
\renewcommand{\epsilon}{\varepsilon}
\newcommand{\rg}[2]{\left \{ #1,\dots,#2  \right \}}

\newcommand{\raffect}{\xleftarrow{\$}}

\newcommand{\qnorm}[1]{\left \lvert  {#1}  \right \rvert_q} 
\newcommand{\support}[2]{\langle {#2} {\rangle}_{{#1}}}

\newcommand{\F}{\mathbb{F}}

\newcommand{\fq}{\F_q}
\newcommand{\fqm}{\F_{q^m}}
\newcommand{\fql}{\F_{q^\ell}}
\newcommand{\set}[1]{\mathcal{#1}}
\newcommand{\card}[1]{\left \lvert  {#1} \right \rvert }

\newcommand{\GL}{\mathsf{GL}}
\newcommand{\MS}[3]{#3^{#1 \times #2}}

\newcommand{\grassman}[3]{\boldsymbol{\mathrm{Gr}}_{#1}(#2,#3)}

\DeclareMathOperator{\rk}{\mathrm{rank}} 
\DeclareMathOperator{\pub}{pub} 

\newcommand{\word}[1]{\boldsymbol{\mathrm{#1}}}
\newcommand{\av}{\word{a}}
\newcommand{\bv}{\word{b}}
\newcommand{\cv}{\word{c}}

\newcommand{\ev}{\word{e}}

\newcommand{\gv}{\word{g}}
\newcommand{\hv}{\word{h}}

\newcommand{\mv}{\word{m}}

\newcommand{\sv}{\word{s}}

\newcommand{\vv}{\word{v}}

\newcommand{\xv}{\word{x}}
\newcommand{\yv}{\word{y}}

\newcommand{\lambdav}{\word{\lambda}}

\newcommand{\mat}[1]{\boldsymbol{\mathrm{#1}}}

\newcommand{\Gm}{\mat{G}}
\newcommand{ \Hm}{\mat{H}}
\renewcommand{\Im}{\mat{I}}
\newcommand{\Mm}{\mat{M}}

\newcommand{\Pm}{\mat{P}}
\newcommand{\Qm}{\mat{Q}}

\newcommand{\Sm}{\mat{S}}
\newcommand{\Tm}{\mat{T}}

\newcommand{\Xm}{\mat{X}}

\newcommand{\ZZ}{\mat{0}}
\newcommand{\Deltam}{\mat{\Delta}}

\newcommand{\CC}{\mathcal{C}}
\newcommand{\VC}{\mathcal{V}}

\newcommand{\FC}{\mathcal{F}}

\newcommand{\gab}{\mathcal{G}ab}

\begin{document}
\title{Security Assessment  of the LG-Cryptosystem}	
\author{ \'Etienne Burle\inst{1} \and Herv\'e Tal\'e Kalachi\inst{2} \and Freddy Lende Metouke\inst{2} \and Ayoub Otmani\inst{1} }

\institute{Univ Rouen Normandie, 
Normandie Univ, LITIS UR 4108, F-76000 Rouen, France
\email{\{Etienne.Burle,Ayoub.Otmani\}@univ-rouen.fr}\\
\and 
University of Yaounde I, Yaounde, Cameroon\\
\email{\{metoukefreddy,hervekalachi\}@gmail.com}
}

\maketitle

\begin{abstract}
The LG cryptosystem is a public-key encryption scheme in the rank metric using the recent family of $\lambdav-$Gabidulin codes and introduced in $2019$ by Lau and Tan. In this paper, we present a cryptanalysis showing that the security of several parameters of the scheme  have been overestimated. We also show the existence of some weak keys allowing  an attacker to find in polynomial time an alternative private key.  

\keywords{Rank-metric codes \and LG cryptosystem \and Cryptanalysis \and Structural attack \and Weak keys}

\end{abstract}

\section{Introduction}
The first  cryptosystem based on linear codes was proposed by McEliece \cite{M78} in 1978, based on Goppa codes equipped with the Hamming metric. 
But, the sizes of its public key are very large and constitute a disadvantage. Over the years, a fairly active line of research has therefore aimed to suggest alternatives to the McEliece cryptosystem with reduced key sizes. Among the most notable alternatives is the idea of using linear codes with another metric, notably the rank metric.

Rank-metric codes were introduced by Delsarte \cite{D78} in 1978 followed by Gabidulin  \cite{G85} with the definition of Gabidulin codes that was used in \cite{GPT91} to construct the first cryptosystem based on rank-metric codes, and called the GPT cryptosystem. The  scheme can be seen as an analogue of the McEliece's one but using Gabidulin codes in place of Goppa codes \cite{M78}. But
the GPT cryptosystem has experienced a succession of attacks and  repairs.
The first attacks were presented by Gibson in \cite{G95,G96}, then by Overbeck's attacks \cite{O05a,O05,O08} and other related work \cite{HMR18,K22,OTN18} that dismantled almost all the existing variants of the GPT cryptosystem \cite{GRH09,G08,RGH10,RGH11}. All these attacks can be justified by the fact that Gabidulin codes are very structured, and therefore their use  always results in a public key that leaks sensitive information. 

Other families of rank metric codes have thus emerged this last decade, with the aim to replace Gabidulin codes with less structured rank metric codes. The first rank metric codes which came to replace Gabidulin codes are the family of Low Rank Parity Check codes introduced in 2013 by Gaborit, Murat, Ruatta and Zémor \cite{GMRZ13} with applications to cryptography. Their work was recently followed by the work of Lau and Tan \cite{LT19} introducing a new family of rank-metric codes called $\lambdav-$Gabidulin codes. Similar to Generalized Reed-Solomon codes \cite{SX11}, a $\lambdav-$Gabidulin code $\gab_{\lambdav}(\gv,k)$ is generated by a matrix that is a product of a generator matrix of the Gabidulin code $\gab(\gv,k)$ with a diagonal matrix. Lau and Tan presented an efficient decoding algorithm for $\gab_{\lambdav}(\gv,k)$ based on the decoding algorithm of $\gab(\gv,k)$. They also proposed the so-called LG cryptosystem using $\lambdav-$Gabidulin codes and claimed that their system is resistant to all known structural attacks on Gabidulin-based cryptosystems such as Overbeck's attack \cite{O07} and Frobenius weak attack \cite{HMR16}. 

In this paper, we exploit the structure and the properties of Gabidulin codes to present a new structural attack on the LG cryptosystem. While the complexity of our attack is not polynomial, it appears that several parameters of the LG cryptosystem do not have the claimed security. We also show the existence of some keys deemed weak allowing the attacker find an alternative private key in polynomial time. At the core of our security analysis is our observation that the LG cryptosystem  hides the structure of a Gabidulin code by means of an invertible matrix whose inverse lies in $\GL_n(\set{V})$ where $\set{V}$ is a  $\fq-$vector subspace of $\fqm$ of dimension at most $3$. Hence, the LG cryptosystem is a particular instance of the Loidreau's cryptosystem \cite{L17} which masks Gabidulin codes by multiplying a generator matrix with the inverse of a homogeneous matrix whose rank weight is extremely low. We then apply a similar technique first proposed in \cite{BL23} to improve the best known attacks on the LG cryptosystem.

The rest of the paper is organized as follows. In Section \ref{sec:def}, we give the notation, properties and mathematical notions that will be needed throughout the paper. We particularly recall the definitions of $\lambdav-$Gabidulin codes and the description of the LG cryptosystem. In Section \ref{sec:cryptanalysis}, we present the new cryptanalysis of the LG cryptosystem and also present our attack on some keys that can be viewed as weak in Section \ref{sec:weak}. 

\section{Notation and Preliminaries} \label{sec:def}

The  symbol  $\eqdef$ is used to define  the left-hand side object. $\card{\set{S}}$ defines the cardinality of a set $\set{S}$.
We shall write $x \raffect \set{S}$ to express that $x $ is sampled  according to the uniform distribution over  a set  $\set{S}$.  
The finite field with $q$ elements where $q$ is a power of a prime number 
 is written as $\fq$.
All vectors will be regarded by default as row
vectors and  denoted by boldface letters like $\av = (a_1,\dots{},a_n)$. 
The linear space over a field $\fq$  spanned by vectors $\bv_1,\dots,\bv_k$   belonging 
to a vector space over a field containing $\fq$ is written as $\support{\fq}{\bv_1,\dots,\bv_k}$.
 The set of $r \times n$ matrices with entries in a set $\set{V} \subseteq \fq$ is denoted by $\MS{r}{n}{\set{V}}$ and
$\GL_n(\set{V})$ is the subset of $\MS{n}{n}{\set{V}}$  of invertible matrices.
The identity matrix of size $n$ is written as $\Im_n$. 
The transpose operation  is denoted by the symbol $^{\mathsf{T}}$. 
  
\subsection{Rank Metric} We consider a finite field extension $\fqm/\fq$ of degree $m \geq 1$ where $q$ is a power of a prime number.  
The \emph{support}   
 of   $\xv$ from $\fqm^L$ denoted by $\support{\fq}{\xv}$ is  the vector subspace  over $\fq$ spanned by its entries, namely $\support{\fq}{\xv} \eqdef \support{\fq}{x_1,\dots{},x_L} \subseteq \fqm$. 
 The   \emph{rank weight} of $\xv$ is  $\qnorm{\xv} \eqdef \dim \support{\fq}{\xv}$.
 Likewise, the  support $\support{\fq}{\Xm}$ of a  matrix $\Xm = (x_{i,j})$ 
 is  the vector subspace over $\fq$ spanned by all its  entries, 
 and its weight is $\qnorm{\Xm} \eqdef \dim \support{\fq}{\Xm}$. 
 For $a \in \fqm$ and for any integer $i$, $a^{[i]} \eqdef a^{q^i}$.
 Lastly, we let  $\grassman{w}{q}{m}$ be the set of all $w-$dimensional linear subspaces 
over $\fq$ included in  $\fqm$.

\subsection{$\lambdav-$Gabidulin Codes}\label{section_lambda_Gabidulin_codes}
In \cite{LT19}, a new family of codes called $\lambdav-$Gabidulin codes was introduced by Lau and Tan. These codes can be seen as generalized Gabidulin codes \cite{G91} by similarity to the well known family of generalized Reed-Solomon codes \cite{SX11}. 
\begin{definition}[$\lambdav-$Gabidulin codes]
    Given two integers $k$ and $n$ such that $k \leq n$, let $\gv = (g_1,\ldots,g_n) \in  \fqm^n$ such that $\qnorm{\gv}=n$, and $\lambdav=(\lambda_1,\ldots,\lambda_n) \in \fqm^n$. The $\lambdav-$Gabidulin code $\gab_{\lambdav}(\gv,k)$ of dimension $k$ associated to $\gv$ and $\lambdav$ is the code defined by the generator matrix $\Gm_{\lambdav}  =\Gm \Deltam $ where
    $\Deltam = (\Delta_{i,j} )$ is the diagonal matrix such that $\Delta_{i,i}=\lambda_i$ and $ \Gm$ is defined as follows
\begin{equation} \label{def:G_Gab}
    \Gm = \begin{pmatrix}
    g_1            &  &\cdots&  & g_n          \\
    g_1^{[1]}      &  &\cdots&  & g_n^{[1]}   \\
    \vdots         &  &      &  & \vdots   \\
    g_1^{[k - 1]}  &  &\cdots&  & g_n^{[k - 1]} 
\end{pmatrix}
\end{equation}
In the specific case where $\lambdav = (1, \dots, 1)$ then  $\gab_{\lambdav}(\gv,k) = \gab(\gv,k)$  is the Gabidulin code of dimension $k$ and support $\gv$.
\end{definition}
\begin{remark}
A matrix of the form of \eqref{def:G_Gab} is called  a Moore matrix generated by $\gv$. 
\end{remark}
As in the case of Generalized Reed-Solomon codes, the decoding algorithm of $\gab_{\lambdav}(\gv,k)$ is based on the decoding algorithm of $\gab(\gv,k)$ with the difference that $\lambdav-$Gabidulin codes come to a loss in terms of error correction capacity. Indeed, given a corrupted codeword $\yv = \cv + \ev$ with $\cv$ belonging to  $\gab_{\lambdav}(\gv,k)$ and $\ev$ being an error vector, there exists $\mv \in \fqm^k$ such that $\yv = \mv \Gm \Deltam + \ev$. So, $\yv \Deltam^{-1} = \mv \Gm  + \ev \Deltam^{-1}$ and, applying a decoding algorithm of $\gab(\gv,k)$ to $\yv \Deltam^{-1}$ allows the recovering of $\mv$ whenever $\qnorm{\ev \Deltam^{-1}}$ is less than $t \eqdef \frac{n-k}{2}$. Since $\qnorm{\ev \Deltam^{-1}} \leq \qnorm{\ev} \qnorm{\Deltam^{-1}}$, the previous condition is satisfied whenever $\qnorm{\ev} \leq \dfrac{t}{\qnorm{\Deltam^{-1}}}$.

Like Gabidulin codes, we also have the following proposition from \cite{LT19}, stating that the dual of a $\lambdav-$Gabidulin code is also a $\lambdav-$Gabidulin code.
\begin{proposition}
Given a $\lambdav-$Gabidulin code $\gab_{\lambdav}(\gv,k)$, there exists\\ $\hv = (h_1,\ldots, h_n) \in \fqm^n$ such that $\qnorm{\hv} = n$ and $\gab_{\lambdav}(\gv,k)^\perp = \gab_{\lambdav^{-1}}(\hv,n-k)$ with $\lambdav^{-1} \eqdef (\lambda_1^{-1},\ldots,\lambda_n^{-1})$.   
\end{proposition}

\subsection{LG Encryption Scheme} \label{section_lambda_gabidulin_scheme}
    The LG cryptosystem is a McEliece-like cryptosystem using $\lambdav-$Gabidulin codes. The system can be described as follows.
    \begin{itemize}
     \item \textbf{Parameters.}   Choose integers $m,n,k$ such that $m \geq n>k$, and $k$ does not divide $n-1$, and define $t \eqdef \left\lfloor {\dfrac{{n - k}}{2}} \right\rfloor $ and $a \eqdef \left\lfloor {\dfrac{{t}}{3}} \right\rfloor$

    \item \textbf{Key Generation.} 
    \begin{itemize}
     \item Generate uniformly at random $\gamma \in \fqm \backslash \fq$ such that $\gamma^2\neq 1$, and set $\lambdav \eqdef (\lambda_1, \ldots, \lambda_n)$ with $\lambda_i \raffect \{ \gamma,\gamma^{-1} \}$ for each $i$ in $\rg{1}{n}$
    \item Choose $\gv \in \fqm^n$ so that $\qnorm{\gv} = n$ and construct the generator matrix $\Gm_{\lambdav}$ of $\gab_{\lambdav}(\gv,k)$
     \item Pick uniformly at random  $\Sm \raffect \GL_k(\fqm)$ and an $n\times n$ invertible matrix $\Pm$ with entries from $\{ b\gamma,c\gamma^{-1} \mid b,c \in \fq \}$.
    The public key is then  $(\Gm_{\pub} \eqdef \Sm \Gm_{\lambdav} \Pm^{-1}, a)$, and the private key is $(\Sm, \gv, \lambdav, \Pm)$.
    \end{itemize}
    
    \item \textbf{Encryption.} Given a plaintext $\mv \in \fqm^k$, pick at random a vector $\ev \in \fqm^n$ such that $\qnorm{\ev}=a$, and compute the ciphertext $\yv \eqdef \mv \Gm_{\pub} + \ev $.
    
	\item \textbf{Decryption.} Let $\Deltam$ be the diagonal matrix with entries $\Delta_{ii}=\lambda_i$ for $i$ in $\rg{1}{n}$. Compute $\yv' = \yv \Pm \Delta^{-1}$ and apply to $\yv'$ the decoding algorithm of  $\gab(\gv,k)$  in order to recover $ \mv \Sm$. The plaintext $\mv$ is then recovered by multiplying  by $\Sm^{-1}$.
	\end{itemize}
Based on the above description, we have the following proposition.
\begin{proposition}\label{Prop:LG_to_Loidreau}
   The public matrix $\Gm_{\pub}$ of the LG cryptosystem can be rewritten as   
    \[
    \Gm_{\pub}=\Sm \Gm \Qm^{-1}
    \] 
    where $\Gm$ is the generator matrix of a Gabidulin code and the entries of the $n \times n$ matrix $\Qm$ belonging to $\set{V}=\left\langle {1 ,\gamma ^{ - 2} ,\gamma ^2 } \right\rangle _{\F_q }$.
\end{proposition}		

\begin{proof}
    We have 
     $\Gm_{\pub}=\Sm \Gm_{\lambdav} \Pm^{-1} = \Sm \Gm \Deltam \Pm^{-1}$ where $\Sm \in \GL_k(\fqm)$, $\Gm$ is a generator matrix of $\gab(\gv,k)$ and the entries of $\Pm$ belong to $\{b\gamma,c\gamma^{-1} \mid b,c\in \fq \}$, while $\Deltam^{-1}$ is a diagonal matrix with entries $\Delta_{ii}=\lambda_i \in \{ \gamma, \gamma^{-1} \}$ for $i$ in $\rg{1}{n}$. So, the matrix $\Qm = \Pm \Deltam^{-1}$ has its entries in $\left\langle {1 ,\gamma^{-2} ,\gamma ^2 } \right\rangle _{\fq }=\set{V}$.
\qed\end{proof}

\begin{remark}\label{rq:LG_to_Loidreau}
The previous proposition shows that the LG cryptosystem hides the structure of a Gabidulin code by means of an invertible matrix whose inverse lies in $\GL_n(\set{V})$ where $\set{V}$ is a  $\fq-$vector subspace of $\fqm$ of dimension at most $3$. Hence, the LG cryptosystem is a particular instance of the Loidreau's cryptosystem \cite{L17}.
\end{remark}

\begin{remark} \label{rem:basis_of_Q}
For any non-zero $\alpha$ from $\fqm$, we also have $\Gm_{\pub}= \left(\alpha\Sm \right) \Gm \left(\alpha \Qm \right)^{-1}$ which entails that the entries of the matrix $\Qm$ from Proposition \ref{Prop:LG_to_Loidreau} can be considered to belong to 
$\left\langle \alpha ,\alpha \gamma^{-2} , \alpha \gamma^{2} \right \rangle _{\fq }$. 
\end{remark}

\section{A Novel Cryptanalysis of the LG Cryptosystem} \label{sec:cryptanalysis}

In this section, we propose a new attack that breaks most of the proposed parameters. But before describing the attack, we collect in the following some preliminary results. 



\begin{lemma}\label{FreeChoiceOfsupport}
    Let $n, m$ be positive integers such that $n\leq m$ and $\gv=(g_1,\dots,g_n) \in \F_{q^m}^n$ such that $\qnorm{\gv}=n$. 
    \begin{enumerate}
        \item For any $\alpha \in \fqm^*$, $\gab(\alpha \gv,k) = \gab(\gv,k)$.
        \item Given $\gv'=(g'_1,\dots,g'_m) \in \F_{q^m}^m$ such that $\qnorm{\gv'} =m$, there exist at least $q^m-1$ full rank matrices $\Tm \in \F_{q}^{m\times n}$ such that $\gab( \gv' \Tm , k) = \gab(\gv,k)$. 
    \end{enumerate}
\end{lemma}

\begin{proof}
For the first point, we can remark that a generator matrix of $\gab(\alpha \gv,k)$ is given by the following expression
\[ 
    \Gm_\alpha = 
    \begin{pmatrix}
    \alpha   &     0      &\cdots & 0          \\
      0      & \alpha^{[1]} &\cdots & 0   \\
     \vdots  &            &  \ddots  & \vdots   \\
       0     & 0 &\cdots & \alpha^{[k-1]} 
\end{pmatrix}
    \begin{pmatrix}
    g_1            &  &\cdots&  & g_n          \\
    g_1^{[1]}      &  &\cdots&  & g_n^{[1]}   \\
    \vdots         &  &      &  & \vdots   \\
    g_1^{[k - 1]}  &  &\cdots&  & g_n^{[k - 1]} 
\end{pmatrix}
\]
This clearly shows that $\gab(\alpha \gv,k) = \gab(\gv,k)$.
 Now we treat the second point, and let us consider an arbitrary $\gv'$ in $\F_{q^m}^m$ such that $\qnorm{\gv'} =m$. Notice that the entries of $\gv'$ form a basis of $\fqm$, and consequently
for every $\alpha$ in $\fqm^*$, as $\qnorm{\gv'} = m \geq n = \qnorm{\alpha \gv} = \qnorm{\gv}$, there exists a full rank matrix $\Tm_\alpha \in \fq^{m \times n}$ such that $\gv' \Tm_\alpha = \alpha \gv$. To finish, given $\alpha$ and $\alpha^\prime$ in $\fqm^*$, one can easily establish that having $\Tm_\alpha = \Tm_{\alpha^\prime}$ results in $\alpha = \alpha^\prime$. 
%
\qed
\end{proof}

\begin{theorem}\label{Theo:number_of_pot_solution}
    Let $\CC_{\pub}$ be the public code of the LG cryptosystem and $\hv_0 \in \fqm^m$ an arbitrary vector of weight $\qnorm{\hv_0} = m$. Then, there exist at least $q^m-1$ full rank matrices $\Mm \in \set{V}^{m \times n}$ such that $\Hm_0 \Mm$ is a parity-check matrix of $\CC_{\pub}$ where $\set{V} = \support{\fq}{1,\gamma^2, \gamma^{-2}}$ and $\Hm_0$ is  $(n-k) \times m$ parity-check matrix of  $\gab(\hv_0,k)$.
\end{theorem}
\begin{proof}
From Lemma~\ref{FreeChoiceOfsupport}, there exist at least $q^m-1$ full-rank matrices $\Tm \in \F_{q}^{m\times n} $ such that $\Hm_0 \Tm$ is a parity check matrix of the hidden Gabidulin code of the LG cryptosystem (see Proposition~\ref{Prop:LG_to_Loidreau}).  
Therefore,  $\Hm_0 \Tm \Qm^{\mathsf{T}}$ is a parity check matrix of $\CC_{\pub}$ where from Proposition~\ref{Prop:LG_to_Loidreau} we know that $\Qm^{\mathsf{T}}$ belongs to $\GL_n(\VC)$ which entails that the entries of $\Mm = \Tm \Qm^{\mathsf{T}}$ are also elements of $\set{V}$. \qed
\end{proof}

The starting point of the attack is  to choose an arbitrary vector $\hv_0 =(h_{0,1},\dots,h_{0,m}) \in \fqm^m$ of rank weight $\qnorm{\hv_0} = m$ and set up the following linear system 
\begin{equation}\label{Key_Eq_1}
    \Gm_{\pub} \Mm^{\mathsf{T}} \Hm_0^{\mathsf{T}} = \ZZ
\end{equation}
 The linear system has $k(n-k)$ equations and $mn$ unknowns from $\fqm$ which are the entries of the matrix $\Mm$. In order to solve such a system, we adopt the combinatorial approach given in \cite{BL23} that exploits the fact that the entries of $\Mm$ lie  in a small $\fq-$subspace $\FC \subsetneq \fqm$ of dimension $r \ge 3$. 

\begin{proposition}\label{attack_basis}
 Given $\FC \in \grassman{r}{q}{m}$ with $3 \leq r \leq k-\lceil k^2/n \rceil$, if there exists $\alpha \in \fqm^*$ such that $\alpha \cdot \VC \subseteq \FC$, then an alternative private key of the LG cryptosystem can be recovered by solving a linear system with $mk(n-k)$ equations and $mnr$ unknowns over $\F_{q}$.
\end{proposition}

\begin{proof}
   
Let us consider a basis $\{f_1,\dots,f_r\}$ of $\FC$, and let us define for each entry $m_{i,j}$ of $\Mm^{\mathsf{T}}$,  
$r$ unknowns $x_{i,j,1},\dots,x_{i,j,r}$ that lie in $\fq$ such that we have $m_{i,j} = \sum_{\ell = 1}^r x_{i,j,\ell} f_\ell$. Next, by replacing $m_{i,j}$ with this expression in  \eqref{Key_Eq_1}, we obtain the following linear system,  
\begin{equation}\label{SystemOverFqm}
\sum_{i=1}^n g^{\pub}_{a,i} 
\sum_{j=1}^m
\left( \sum_{\ell = 1}^r x_{i,j,\ell} f_\ell \right)
h_{0,j}^{q^b}
=0
\end{equation}
for  $a$ in $\rg{1}{k}$, $b$ in $\rg{0}{n-k-1}$, and
 $g^{\pub}_{a,i}$ being the entry of $G_{\pub}$ located at $(a,i)$.
The linear system \eqref{SystemOverFqm} has $nrm$ unknowns and $k(n-k)$ equations, and can be unfolded over $\F_{q}$ resulting into a linear system with $mk(n-k)$ equations. 
We want more equations than unknowns, so we have the condition $nrm\leq mk(n-k)$, that is to say $r \leq k- \lceil k^2/n \rceil$. Given a non-zero solution,  one can then calculate $ \Mm^\prime = \alpha \Mm$ which together with $\Hm_0$ can serve as an alternative secret key. Indeed, given a ciphertext $\yv = \mv \Gm_{\pub} + \ev$ with $\mv \in \F_{q^m}^k$, and $\ev \in \F_{q^m}^n$ such that $\qnorm{\ev}= \lfloor \lfloor\frac{n-k}{2}\rfloor / 3 \rfloor$, we have $\Hm_0 \Mm^{\prime} \yv^{\mathsf{T}}=\Hm_0 \Mm^{\prime} \ev^{\mathsf{T}}=\sv^{\mathsf{T}}$. 
$\Mm^\prime \ev^{\mathsf{T}}$ can then be recovered in polynomial time using the parity check matrix $\Hm_0$. To finish, given $\Mm^\prime \ev^{\mathsf{T}}$ and $\Mm^\prime$, $\ev$ is easily recovered using linear algebra.\qed
\end{proof}
According to the previous Proposition \ref{attack_basis}, one can devise the attack that finds an alternative secret key of the LG cryptosystem in two main steps: 
\paragraph{\bf Step 1 -- Guessing $\FC$.} As stated in Proposition \ref{attack_basis}, we are looking for any subspace $\FC$ of $\fqm$ such that $\alpha \cdot \VC \subseteq \FC$ for some non-zero $\alpha \in \F_{q^m}$. Notice that we are exactly in the context of \cite[Section B]{AGHT18} and, from  \cite[Proposition III.1]{AGHT18}, with high probability there are exactly $S = \frac{q^m - 1}{q-1}$ subspaces of the form $\alpha \cdot V$, $\alpha \in \mathbb{F}_{q^m}^*$. From \cite{AGHT18} again, for a fixed $\alpha \cdot \VC$, the probability $P_\alpha $ that a given $\FC$ of dimension $r$ contains $\alpha \cdot \VC$ is $$P_\alpha = \frac{ \dbinom{r}{3}_q}{ \dbinom{m}{3}_q} \approx q^{-3(m-r)}.$$
So, the probability that $\FC$ of dimension $r$ contains a subspace of the form $\alpha \cdot \VC$ can be approximated by $$S P_\alpha \approx q^{m-3m+3r} \approx q^{-2m+3r}.$$ 
%
However, one can also exploit the fact that the vector subspace $\VC$ has a basis of the form $\{1, \beta^2, \beta^4 \}$ for some $\beta \in \fqm \setminus \fq$ (use Remark~\ref{rem:basis_of_Q} with $\alpha = \gamma^2$). Consequently, it is enough to look for $\FC$ among subspaces generated by a basis of the form $\{1, \beta, \beta^2, \cdots, \beta^{r-1} \}$ (assuming $r\ge 7$) for some $\beta \in \fqm$. So, instead of sampling $\FC$ at random in $\grassman{r}{q}{m}$, we simply pick  at random $\beta$ in $\fqm \setminus \fq $ and set $\FC = \FC_\beta  \eqdef \support{\fq}{ 1, \beta, \beta^2, \cdots, \beta^{r-1}}$. We can distinguish at least four cases, all resulting in a successful choice for $\FC_\beta$:
\[
\beta =
\begin{cases}
\gamma  &\text{ then }    \VC = \support{\fq}{1, \gamma^2, \gamma^4} \subset \FC_\beta \\
\gamma^2  &\text{ then }  \gamma^2 \cdot \VC = \support{\fq}{\gamma^2, \gamma^4, \gamma^6} \subset \FC_\beta \\
\gamma^{-1} &\text{ then }  \gamma^{-4} \cdot \VC = \support{\fq}{1, \gamma^{-2}, \gamma^{-4}} \subset \FC_\beta\\
\gamma^{-2} &\text{ then } \gamma^{-6} \cdot \VC = \support{\fq}{\gamma^{-2}, \gamma^{-4}, \gamma^{-6}} \subset \FC_\beta
\end{cases}
\]
All these permit to establish that the success probability is at least $4/(q^m-q)$.
Consequently, this shows that the success probability of Step~1 is at most  
$\max \left\{q^{-2m +3r},4/(q^m-q)\right\}$.

%
%
%

\paragraph{\bf Step 2 -- Solving a linear system.} The idea here is to solve by Linear Algebra  the linear system \eqref{SystemOverFqm} after unfolding it over $\fq$.
 If this results to a trivial solution then we go back to the first step until  a non-zero solution of  \eqref{SystemOverFqm} is obtained, so that we can reconstruct the matrix $\Mm$ in \eqref{Key_Eq_1}. In our simulations, it appears that for a good choice of $\FC$, we end up with a solution space of dimension $m$ and, thanks to Theorem \ref{Theo:number_of_pot_solution}, any non-zero element is valid for the reconstruction of $\Mm$.   

%
%
%
%
%
%
%
%
%

%


While the cost in the second step is known as the cost of solving a linear system, the cost of the first step depends on the guessing method used, as it will be the inverse of the corresponding success probability. A simple analysis  shows that we always have $q^{-2m +3r} < 4/(q^m-q)$ since $m \ge n > r$.  This leads to the following proposition.

\begin{proposition}
    
    Given a public matrix $\Gm_{\pub}$ of the LG cryptosystem, an alternative private key can be recovered with on average  $\mathcal{O}\left(m^3k^3(n-k)^3 q^m\right)$
    operations in $\fq$.

\end{proposition}

\begin{proof}
 The time complexity is given by the fact that on average one has to repeat $ \frac{q^m-q}{4} = \mathcal{O}(q^m)$ times the solving of the linear system \eqref{SystemOverFqm} whose cost is $\mathcal{O}(m^3k^3(n-k)^3)$ operations.
 \qed
\end{proof}

Table~\ref{tab:comlexity} gives the time complexities of our attack on the parameters of the LG cryptosystem. While the attack is still exponential, it appears that several parameters of the LG cryptosystem are broken. These show that our attack should be considered in the future when setting up new parameters of the LG cryptosystem and similar cryptosystems.

 \begin{table}[!hbtp]
\begin{center}
    \begin{tabular}{lcrr}
      \toprule
      \textbf{Scheme} & $(q,m,n,k)$ & \textbf{Claimed security} & \textbf{Complexity} \\
      \midrule
      LG-I & (2, 83, 79, 31) & 128 & 132 \\
      \midrule
      LG-II & (2, 85, 83, 29) & 128 & 134 \\
      \midrule
      LG-III & (2, 97, 89, 23) & 128 & 146 \\ 
      \midrule
      \textbf{LG-IV} &  \textbf{(2, 117, 115, 49)} & \textbf{256} & \textbf{170}\\ 
      \midrule
      \textbf{LG-V} & \textbf{(2, 129, 127, 36)} & \textbf{256} & \textbf{183} \\ 
      \midrule
      \textbf{LG-VI} & \textbf{(2, 133, 131, 34)}  & \textbf{256} & \textbf{187} \\ 
      \midrule
      \textbf{LG-VII} & \textbf{(2, 85, 83, 35)}   & \textbf{140} & \textbf{134} \\ 
      \midrule
      LG-VIII & (2, 91, 89, 28)  & 140 & 140 \\ 
      \bottomrule \\
    \end{tabular}
 \caption{Time complexity of our attack on the parameters of the LG cryptosystem.}
    \label{tab:comlexity}
\end{center}
    \end{table}

 \section{Weak Keys} \label{sec:weak}
 We show in this section that if the integer $m$ is not prime, some choices of secret parameters can lead to a polynomial-time attack. We refer to such parameters as weak keys. Given a public matrix, we show how to detect such weak keys. Then, we describe the polynomial attack allowing to find an alternative secret key from such a public matrix.
Given a vector $\vv \in \fqm^n$ and a positive integer $i$, here we denote by $\vv^{[i]}$ the vector $\vv^{[i]} \eqdef (v_1^{[i]},\dots,v_n^{[i]})$. Similarly, we extend this notation to matrices. For a linear code $\CC$ with generator matrix $\Gm$, we denote by $\CC^{[i]}$ the linear code generated by $\Gm^{[i]}$. Assuming that $m$ is not prime, the extension field $\fqm$ admits proper subfields of the form $\fql$ with $\ell$ dividing $m$. As the parameter  $\gamma$ of the LG cryptosystem is chosen at random in $\fqm \setminus \fq$ there is some chance that $\gamma$ belongs to a proper subfield $\fql$ of $\fqm$. So, let us start by the following Lemma that helps us to know whether $\gamma$ lies in $\fql$ or not. 

\begin{lemma}\label{Gamma_distinguisher}
    Let $m$, $k$, $n$, $\ell$ be positive integers such that  $\ell < \min \{k, n-k\}$ and $\ell$ divides $m$. Given a public code $\CC_{\pub}$ of the LG cryptosystem generated from a secret parameter $\gamma \in \fqm$, we have the following properties:  
 \begin{enumerate}
    \item If $\gamma$ belongs to the subfield $\fql \subsetneq \fqm$ then $\dim \big(\CC_{\pub}+\CC_{\pub}^{[\ell]}\big)= k+ \ell$ 
    \item If $\gamma \notin \fql$ then $\dim \big(\CC_{\pub}+\CC_{\pub}^{[\ell]}\big)= \min \{ 2k, n \}$ with high probability.  
 \end{enumerate}   
\end{lemma}

\begin{proof}
 From Proposition \ref{Prop:LG_to_Loidreau}, $\Gm_{\pub} = \Sm \Gm \Qm^{-1}$ with $\Sm$ in $\GL_k(\fqm)$, $\Gm \in 
\MS{k}{n}{\fqm}$ a Moore matrix and $\Qm \in \GL_n(\VC)$ where $\VC = \support{\fq}{1, \gamma^{-2}, \gamma^2} .$  If $\gamma$ belongs to $\fql$, then we have $\VC \subset \fql$ and $\Qm^{-1} \in \GL_n(\fql).$  As a consequence, we also have $(\Qm^{-1})^{[\ell]} = \Qm^{-1}.$ Thus, since $\ell \leq k-1$, $\dim \big(\CC_{\pub}+\CC_{\pub}^{[\ell]}\big)$ is given by 
\[
\rk \left( \begin{bmatrix}
    \Gm_{\pub} \\
    \Gm_{\pub}^{[\ell]}
\end{bmatrix} \right)
= \rk \left( \begin{bmatrix}
 \Sm & \ZZ \\
 \ZZ & ~\Sm^{[\ell]} \end{bmatrix} \begin{bmatrix}
    \Gm \\
    \Gm^{[\ell]}
\end{bmatrix} \Qm^{-1} \right) = \rk \left(  \begin{bmatrix}
    \Gm \\
    \Gm^{[\ell]}
\end{bmatrix}  \right) = k + \ell.
\]
Now, let us assume that  $\gamma$ is not in $\fql$. Then $\Qm^{-1}$ behaves like a random matrix in $\GL_n(\fqm)$ and it is also the case for $\Gm_{\pub}$. Thus, by \cite[Lemma 5.2]{O08}, we have $\dim \big(\CC_{\pub}+\CC_{\pub}^{[\ell]}\big)= \min \{ 2k, n \}$ with a probability greater than $1-4q^{-m/\ell}$.
\qed
\end{proof}

From the previous Lemma~\ref{Gamma_distinguisher}, one can determine with high probability whether the parameter $\gamma$ was taken in  $\fql$ or not.  Therefore, assuming that  $\gamma$ belongs to $\fql$, the rest of the attack depends on the comparison between $\ell$ and the upper bound $k- \lceil k^2/n \rceil$ of $r$ given in Section~\ref{sec:cryptanalysis} as explained in the following Theorem.

\begin{theorem}\label{weak-keys_attack}
    Let $\CC_{\pub}$ be a $(n,k)$ public code of an instance of the LG cryptosystem defined over $\fqm$. Assume that $m$ is not prime and let $\ell$ be a factor of $m$ such that $\dim \big(\CC_{\pub}+\CC_{\pub}^{[\ell]}\big)= k+ \ell < n$. Then an alternative private key can be recovered in $T$ operations where
    \[
    T = \begin{cases}
            \mathcal{O}\left(m^3k^3(n-k)^3 \right)  \text{ if } \ell \leq k - \lceil k^2/n \rceil\\
             \mathcal{O}\left(m^3k^3(n-k)^3 \min \Big\{q^\ell, q^{2\ell - 3r}\Big\}\right) \text{ if } k - \lceil k^2/n \rceil < \ell \leq k-1
        \end{cases}
    \]
\end{theorem}
\begin{proof}
    We assume that $\dim \big(\CC_{\pub}+\CC_{\pub}^{[\ell]}\big)= k+ \ell$.
     let us consider first the case $\ell \leq k - \lceil k^2/n \rceil$. We  have then 
     $\ell < \min \{k, n-k\}$ and thanks to Lemma \ref{Gamma_distinguisher},  the secret parameter $\gamma$  belongs to $\fql$. As a consequence, the support $\VC$ of  $\Qm$ is included in $\fql$. 
     Therefore, $\FC = \fql$ is a good  candidate allowing us to skip Step 1 in the attack described in Section \ref{sec:cryptanalysis}. Thus, thanks to Proposition \ref{attack_basis}, one can recover an alternative private key by solving a linear system with $mk(n-k)$ equations and $mn \ell$ unknowns. 
     
    To now, if we have $k - \lceil k^2/n \rceil < \ell \leq k-1$ then $\gamma$ still belongs to $\fql$ but $\fql$ is not a good subspace candidate as $\ell$ is greater than  $r = k - \lceil k^2/n \rceil$. However, $\fqm$ can be replaced by $\fql$ in Step 1. So we are now looking for  $\FC \subset \fql$ with a success probability  given by
    $\max \left \{ q^{-2\ell+3r}, q^{-\ell} \right\}$. 
    \qed
\end{proof}

\begin{remark}
If all proper divisors\footnote{A proper divisor of $m$ denotes any positive divisor of $m$ that is different from $1$ and $m$} $\ell$ of $m$ satisfy the conditions $k+\ell < n$ and $\ell \leq k-1$ then the success probability of the algorithm is  dominated by the probability for the selected $\gamma$ to lie in one of the  subfields $\fql \subsetneq \fqm$. 
\end{remark}

We now consider the different sets of parameters for which $m$ is not prime. Notice that this is the case for all  parameters except LG-I and LG-III.  In Table~\ref{tab:weak} we also give the probability $P_w$ of generating $\gamma$ in a subfield $\fql \subsetneq \fqm$. One can remark that for each divisor $\ell$ of  $m$, the conditions $k+\ell < n$ and $\ell \leq k-1$ are always satisfied except for $m = 129$ and $\ell = 43$ in LG-V for which we have $\ell > k-1$. So the success probability in the case of LG-V is the probability that $\gamma$ lies in $\F_{2^3}$.   

\begin{table}[htbp]

 \centering

    \begin{tabular}{llrrr}
      \toprule
      \textbf{Scheme} & $(q,m,n,k)$ & $r$ & Proper divisors of $m$ & $P_w$ \\
      \midrule
      LG-II & (2, 85, 83, 29) & 18  & 5, 17  & $2^{-68}$ \\
      \midrule
      LG-IV &  (2, 117, 115, 49) & 28  & 3, 9, 13, 39 & $2^{-78}$\\ 
      \midrule
      LG-V & (2, 129, 127, 36) & 25 & 3, 43  & $2^{-126}$ \\ 
      \midrule
      LG-VI & (2, 133, 131, 34) & 25 & 7, 19 & $2^{-114}$\\ 
      \midrule
      LG-VII & (2, 85, 83, 35) & 20 & 5, 17 & $2^{-68}$ \\ 
      \midrule
      LG-VIII & (2, 91, 89, 28) & 19 & 7, 13 & $2^{-78}$\\ 
      \bottomrule \\
    \end{tabular}

\caption{Probability of generating weak keys for LG parameters. } \label{tab:weak}       
        
\end{table}

\section{Conclusion} \label{conclusion}
  In this paper, we have conducted two studies of the security of the LG cryptosystem. The first leads to a structural attack which although exponential, shows that several parameters of the LG cryptosystem have been overestimated. The second analysis demonstrates the existence of weak keys leading to a polynomial-time attack. This last attack can be easily avoided, but the first one cannot be avoided without changing the design of the system. Finally, our work can be potentially further improved, in particular the combinatorial/guessing  attack can be replaced by methods resorting  to solving methods of systems of algebraic equations.

\section*{Acknowledgments}
Etienne Burle was supported by RIN100 program funded by Région Normandie. 
Herve Tale Kalachi acknowledges the UNESCO-TWAS and the German Federal Ministry of Education and Research (BMBF) for the financial support under the SG-NAPI grant number 4500454079.
Ayoub Otmani was supported by the grant ANR-22-PETQ-0008 PQ-TLS funded by Agence Nationale de la Recherche (ANR) within France 2030 program.
\bibliographystyle{spmpsci}
\bibliography{main}

\end{document}